\documentclass{article}

\usepackage{microtype}
\usepackage{graphicx}
\usepackage{svg}
\usepackage{booktabs}
\usepackage{multirow}
\usepackage{orcidlink}
\usepackage{hyperref}

\usepackage[accepted]{icml2026}

\usepackage{amsmath,amssymb,amsthm,mathtools}
\usepackage[capitalize,noabbrev]{cleveref}

\theoremstyle{plain}
\newtheorem{theorem}{Theorem}[section]
\newtheorem{lemma}[theorem]{Lemma}
\newtheorem{proposition}[theorem]{Proposition}
\newtheorem{corollary}[theorem]{Corollary}
\theoremstyle{definition}

\newtheorem{assumption}[theorem]{Assumption}
\theoremstyle{remark}

\usepackage{enumitem}

\newcommand{\FlowFake}{FlowFake}
\newcommand{\norm}[1]{\left\lVert#1\right\rVert}
\newcommand{\calX}{\mathcal{X}}
\newcommand{\calY}{\mathcal{Y}}
\newcommand{\calH}{\mathcal{H}}
\newcommand{\calG}{\mathcal{G}}
\newcommand{\bfh}{\mathbf{h}}
\newcommand{\Win}{\mathbf{W}_{\mathrm{in}}}
\newcommand{\Wrec}{\mathbf{W}_{\mathrm{rec}}}
\newcommand{\gleak}{g_{\mathrm{leak}}}
\newcommand{\Vleak}{\mathbf{V}_{\mathrm{leak}}}
\newcommand{\Cm}{\mathbf{C}_m}
\newcommand{\gell}{g_{\ell}}
\newcommand{\gull}{g_{u}}
\newcommand{\cmin}{c_{\min}}
\newcommand{\cmax}{c_{\max}}

\icmltitlerunning{\ FlowFake: Liquid Networks for Audio Deepfake Detection}

\icmlsetsymbol{equal}{*}

\begin{document}
\twocolumn[
  \icmltitle{FlowFake: Liquid Networks for\\ Audio Deepfake Detection}

  \begin{icmlauthorlist}
    \icmlauthor{Shivaay Dhondiyal\,\orcidlink{0009-0003-4930-1541}}{equal,dtu}\hspace{0.2cm}
    \icmlauthor{Divyansh Sharma\,\orcidlink{0009-0005-5476-2689}}{equal,dtu}\hspace{0.2cm}
    \icmlauthor{Dinesh Kumar Vishwakarma\,\orcidlink{0000-0002-1026-0047}}{dtu}
  \end{icmlauthorlist}

  \icmlaffiliation{dtu}{Delhi Technological University, New Delhi, India}
  
  \vskip 0.3in
]
{\let\thefootnote\relax\footnotetext{\textsuperscript{*}Equal contribution}}
{\let\thefootnote\relax\footnotetext{\textsuperscript{1}Delhi Technological University, New Delhi, India.\\
\\
Accepted to the Workshop on Learning to Listen: Machine Learning for Audio at ICML 2026, Seoul, South Korea. Copyright © 2026 by the author(s).
}}


\begin{abstract}
Audio deepfakes generated by neural text-to-speech and voice-cloning
systems threaten speaker verification and public discourse at scale.
The core challenge is \emph{cross-dataset generalization}: detectors
trained on one synthesis pipeline collapse on unseen forgeries.
We argue that this failure is primarily because of structural synthetic-speech artifacts
 which are multi-timescale \emph{trajectory anomalies}. Though every existing
detector aggregates a fixed-window frame statistics, this misaligns the
architecture with the signal.  We propose \textbf{\FlowFake{}}, a Liquid
Time-Constant (LTC) architecture whose hidden state evolves via a
learned ODE, with per-neuron adaptive time constants simultaneously
resolving spectral ($\sim$10\,ms) and prosodic ($\sim$2\,s) cues.
At only $\approx$34\,K parameters \FlowFake{} achieves formal BIBO
stability and $\mathcal{O}(\Delta t^4)$ integration error.  On a
four-dataset cross-domain benchmark (ASVspoof\,2019-LA, FakeOrReal,
InTheWild, MLAAD), \FlowFake{} reaches \textbf{75.29\%} on
ASVspoof\,2019 trained only on FakeOrReal and \textbf{79.97\%}
trained only on MLAAD. It outperforms RawGAT-ST and Whisper-DF on
every evaluated pair and matching SSL Wav2vec2 (300$\times$ larger)
at 0.01\% of its parameter count.
The source code is available on \href{https://github.com/GhostRider2023/FlowFake.git}{GitHub}.
\end{abstract}

\section{Introduction}
\label{sec:intro}

Neural text-to-speech and voice-cloning systems
\citep{shen2018tacotron2,jia2018transfer} now produce near-human
speech at negligible cost, enabling impersonation attacks against
voice-based authentication \citep{wu2015spoofing,todisco2019asvspoof}
and contributing to high-profile disinformation incidents.  The
practical bottleneck is not raw detection accuracy on a single
synthesis family but \emph{cross-dataset generalization}: a detector
trained on one TTS pipeline must remain effective against unseen
forgeries from different vocoders, languages, and recording
conditions~\citep{muller2022itw,frank2021wavefake}.  This is the
deployment regime that matters; it is also the regime in which
every published detector collapses.

Existing countermeasures fall into three families, all of which
collapse out of distribution.  Graph attention networks
\citep{jung2022rawgat,jung2022aasist} memorise dataset-specific
spectral artifacts: RawGAT-ST trained on FakeOrReal reaches only
$49.1{\pm}18.1\%$ on ASVspoof\,2019, near random chance.
Self-supervised frontends~\citep{tak2022wav2vec} fine-tune
$\sim$300\,M parameter transformers with fixed attention windows;
deployment is prohibitive and the cross-seed variance is large
($\pm17.5$\,pp on MLAAD$\to$ITW, \Cref{tab:full}).  ASR encoder
repurposing~\citep{radford2023whisper,muller2024mlaad} inherits
representations optimised for recognition semantics, not low-level
forgery cues, yielding only 44.9\% on MLAAD when trained on
ASVspoof\,2019.

\textit{Our central hypothesis is that the shared failure mode is
architectural, not data-driven: synthetic-speech artifacts are
trajectory anomalies in how spectro-temporal features evolve over
time, but every existing detector aggregates frame-level statistics
over a fixed context window, structurally erasing the trajectory
information.}  Physical articulation imposes well-characterised
dynamical constraints (vocal tract changes at $\sim$10--100\,ms,
prosodic contours at $\sim$100-2000\,ms). TTS systems violate
these constraints in synthesis paradigm specific ways, producing
\emph{trajectory anomalies} (deviations in how features evolve) and
not in their instantaneous values.  A model endowed with the
correct structural prior (continuous-time trajectory
modelling) should detect synthesis paradigm agnostic violations of
articulatory dynamics rather than the fingerprint of any one synthesis
pipeline.

\paragraph{Contributions.}
We introduce \textbf{\FlowFake{}}, the first Liquid Time-Constant
\citep{hasani2021ltc} architecture for audio deepfake detection.
Our contributions are:
\begin{enumerate}[leftmargin=*,noitemsep,topsep=2pt]
  \item A gradient stable LTC variant with simplified $\tanh$
    synapse and log-parameterized adaptive time constants
    $\tau_i\in[0.05,10]$\,s, integrated by 4th-order Runge-Kutta
    (\Cref{sec:method}).
  \item Formal BIBO stability (\Cref{thm:bibo}) and
    $\mathcal{O}(\Delta t^4)$ RK4 error bound (\Cref{prop:rk4});
    full proofs in \Cref{app:proofs}, where we additionally
    establish gradient
    attenuation (\Cref{prop:grad}) and noise robustness (\Cref{prop:noise}).
  \item Cross-dataset accuracy of 75.29\% (FoR$\to$ASV19) and 79.97\%
    (MLAAD$\to$ASV19) on ASVspoof\,2019 (\Cref{tab:full}),
    outperforming all baselines on the hardest pairs at 0.01\% of
    SSL Wav2vec2's parameter count.
\end{enumerate}

\section{Problem Setup}
\label{sec:setup}

Let $\calX{=}\mathcal{L}^2(\mathbb{R}_{\ge0};\mathbb{R})$ be the
waveform space and $\calY{=}\{0,1\}$ be the label space (0:~bonafide,
1:~synthetic).  A synthesis domain $\mathcal{D}_s$ is a distribution
over $\calX{\times}\calY$ induced by TTS pipeline $s\in\mathcal{S}$.
A detector $f_\theta:\calX\to[0,1]$ has domain-specific risk
$\mathcal{R}_s(f_\theta){=}\mathbb{E}_{(x,y)\sim\mathcal{D}_s}
[\ell(f_\theta(x),y)]$ with $\ell$ the binary cross-entropy.  Our
objective is to minimise the cross-domain gap
$\calG\coloneqq\max_{s\neq s_{\mathrm{tr}}}\mathcal{R}_s(f_\theta)
-\mathcal{R}_{s_{\mathrm{tr}}}(f_\theta)$ without access to any
test-domain data; a strict no target domain adaptation setting.
The hidden state space is $\calH{=}\mathbb{R}^H$ with $H{=}32$.
Full notation appears in \Cref{tab:notation} (\Cref{app:notation}).

\section{\FlowFake{} Architecture}
\label{sec:method}

\begin{figure*}[t]
  \centering
\includegraphics[width=\textwidth]{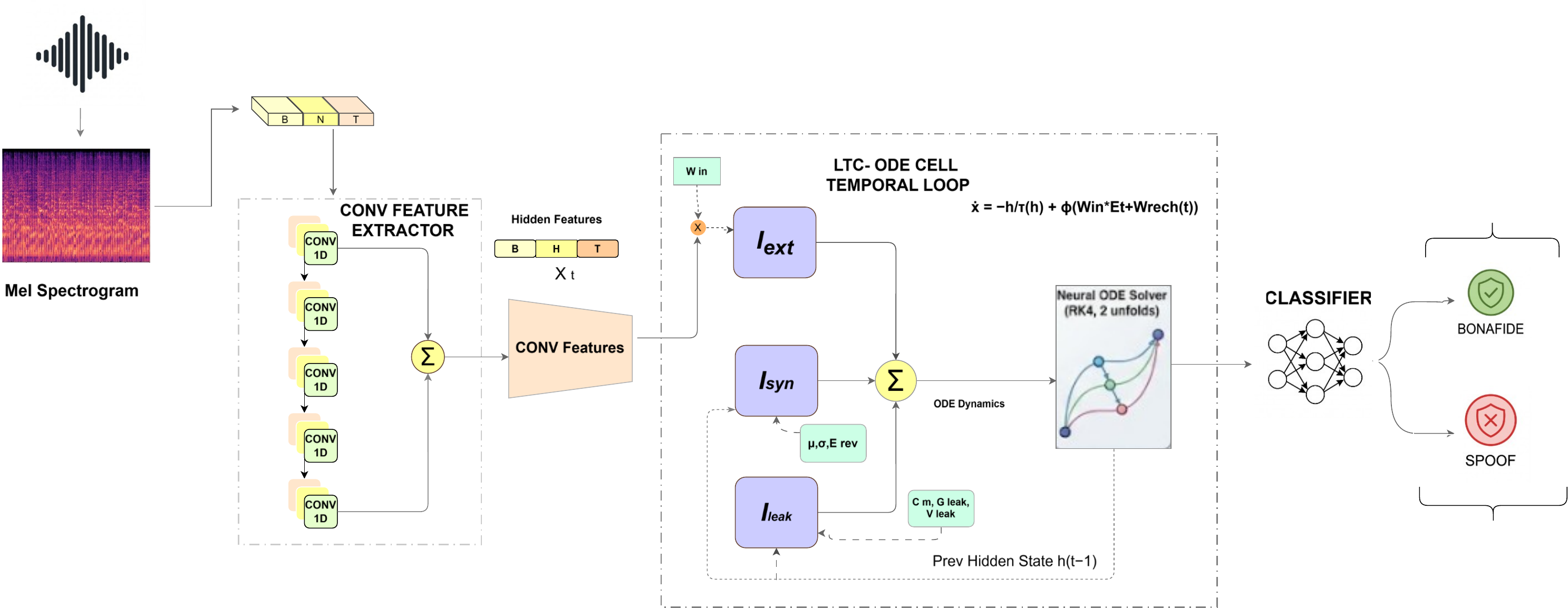}
  \caption{
Overview of the proposed \FlowFake{} framework for audio deepfake detection.
The input audio signals are first converted into Mel-spectrograms
and passed through five Conv1D layers to transform the spectrograms
into $B \times H \times T$. The extracted embeddings are fed into
the LTC-based neural ODE module to extract temporal features through
the following equations:
$I_{in} = W_{in}x_t$,
$I_{syn}(h(t)) = \tanh(W_{rec}h(t))$, and
$I_{leak} = g_{leak}(V - h(t))$.
The ODE module uses the fourth-order Runge--Kutta (RK4) method
with two unfolds to integrate the differential equations.
Finally, a classifier is used to predict whether the audio signals
are bonafide or spoof speech.
}
  \label{fig:arch}
\end{figure*}

\paragraph{Frontend:}
A 16\,kHz mono waveform $w\in\calX$ is $\ell_2$ normalised and
converted into to a log-Mel spectrogram
$\mathbf{X}\in\mathbb{R}^{N\times T}$ with $N{=}128$ bands,
$n_{\mathrm{fft}}{=}512$, hop 160 samples
($\Delta t_{\mathrm{frame}}{\approx}10$\,ms; batch dim is suppressed
for clarity).  The 10\,ms resolution maintains exactly the
spectro-temporal scale at which neural-vocoder artifacts manifest.

\paragraph{Convolutional encoder: }
Five 1D-conv layers (kernels $\{5,1,3,3,1\}$, BN, ReLU, $H{=}32$
output channels) squeezes each frame to an $H$-dimensional embedding
$E_t = \mathrm{CNN}(\mathbf{X})_t\in\mathbb{R}^H$ for
$t=1,\dots,T'{\le}T$, where $T'$ is the dataset specific LTC
coverage length (\Cref{tab:hyper}).  Diversifying kernel sizes capture
spectral structure at multiple resolutions without parameter
inflation.

\paragraph{Liquid Time-Constant cell:}
The hidden state $\bfh(t)\in\calH$ updates  according to the ODE
\begin{align}
  \frac{d\bfh(t)}{dt}
  = \Cm^{-1}\!\odot\!\Bigl[
      &\Win E_t + \tanh\!\bigl(\Wrec\bfh(t)\bigr) \notag\\
      &+ \gleak\!\odot\!\bigl(\Vleak - \bfh(t)\bigr)
    \Bigr],
  \label{eq:ltc}
\end{align}
where $\Cm,\gleak,\Vleak\in\mathbb{R}^H$ are learnable membrane
capacitance, leak conductance, and resting potential, and $\odot$
denotes the Hadamard product.  The three terms play complementary
roles. \textit{Input drive} $\Win E_t$ injects CNN embeddings at
every frame. The \textit{recurrent term} $\tanh(\Wrec\bfh(t))$
provides constrained internal dynamics. The \textit{leak term}
$\gleak\odot(\Vleak-\bfh(t))$ provides a dissipative restoring
force that guarantees BIBO stability (\Cref{thm:bibo}).  We replace
the biological sigmoidal synapse of \citet{hasani2021ltc} with
$\tanh(\Wrec\bfh)$, dropping parameters $\approx3\times$ and
stabilizing gradient norms over 50-200 frame rollouts. The
original synapse brings in non-monotone gradients that empirically
destabilise audio training.

\paragraph{Adaptive time constants:}
Each neuron has a learnable timescale $\tau_i=\exp(\hat\tau_i)$,
parameterized in log-space (guaranteeing positivity without
projection) and clamped to $[0.05,10]$\,s.  \textit{At convergence
the learned $\tau_i$ are empirically bimodal (a fast cluster at
0.1--0.3\,s and a slow cluster at 1.5-5.0\,s) capturing the two
dominant TTS-artefact timescales (spectral-frame and
prosodic-phrase) automatically from data; see
\Cref{app:hyper} for empirical distributions.}

\paragraph{RK4 integration and head:}
\Cref{eq:ltc} is solved via 4th-order Runge-Kutta
\citep{hairer1993solving} at $\Delta t{=}0.01$\,s with $K{=}2$
solver unfolds per audio frame, yielding global error
$\mathcal{O}(\Delta t^4)$ (\Cref{prop:rk4}). The terminal state
$\bfh(T')$ routes to a two layer FC head with weights
$\mathbf{W}_1\in\mathbb{R}^{d\times H}$,
$\mathbf{W}_2\in\mathbb{R}^{1\times d}$ ($d{=}16$) and biases
$\mathbf{b}_1\in\mathbb{R}^d, b_2\in\mathbb{R}$.  Training uses
BCEWithLogitsLoss with positive-class weight
$w_+\!=\!N_{\mathrm{sp}}/N_{\mathrm{bon}}$ to handle class
imbalance (e.g.\ 9:1 in ASVspoof\,2019-LA).

\paragraph{Why LTC is a better structural prior:}
A discrete recurrent or attention model with
$\bfh_t=f(\bfh_{t-1},E_t)$ pools information over a fixed
context window and is structurally insensitive to how features
change between samples. Two clips with identical frame-level
statistics but entirely different trajectories (one natural, one
synthetic) are indistinguishable to such a model.  \textit{The LTC
ODE models $d\bfh/dt$ directly, making the detector structurally
sensitive to trajectory shape (the precise object in which TTS
artifacts manifest).}  The leak conductance further provides a
closed form cross-domain robustness guarantee: noise perturbations
are exponentially suppressed at rate $\gell/\cmin$
(\Cref{eq:noise_bound}, derived in \Cref{app:proofs}), a property
absent in any discrete recurrent model.  Finally, at
$|\theta|{=}34$\,K the KL-divergence term in PAC-Bayes
bounds~\citep{mcallester1999pac} is orders of magnitude smaller
than for 300\,M-parameter SSL models, solidifying generalization
guarantees across unseen domains.

\section{Theoretical Analysis}
\label{sec:theory}

We denote $\gell{\coloneqq}\min_i(\gleak)_i$,
$\gull{\coloneqq}\max_i(\gleak)_i$,
$\cmin{\coloneqq}\min_i(\Cm)_i$,
$\cmax{\coloneqq}\max_i(\Cm)_i$, all strictly positive by
definition.

\begin{assumption}[Bounded Input]\label{ass:bounded}
$\norm{E_t}_2\le M<\infty$ for all $t\ge0$.
\end{assumption}

\begin{theorem}[BIBO stability of \FlowFake{}]\label{thm:bibo}
Under \Cref{ass:bounded} with $(\Cm)_i,(\gleak)_i>0$ elementwise,
the system \Cref{eq:ltc} is BIBO stable: for any
$\bfh(0)\in\calH$, $\norm{\bfh(t)}_2\le R^*$ for all $t\ge t_0$,
where
\begin{equation}
  R^* = \frac{\norm{\Win}_2 M + \sqrt{H}
              + \gull\norm{\Vleak}_2}{\gell}.
  \label{eq:rstar}
\end{equation}
\end{theorem}

\noindent\textit{Proof sketch.}
With $V(\bfh){=}\tfrac12\norm{\bfh}_2^2$, differentiation along
\Cref{eq:ltc}, the Cauchy-Schwarz inequality, the saturation bound
$\norm{\tanh(\mathbf{u})}_2{\le}\sqrt{H}$
(\Cref{lem:tanh}, \Cref{app:proofs}), and elementwise leak bounds
yield $\dot V\le-\norm{\bfh}_2(\beta\norm{\bfh}_2-\alpha)$ with
$\beta=\gell/\cmin>0$.  Hence $\dot V<0$ outside the ball of radius
$R^*=\alpha/\beta$; LaSalle's invariance principle finalizes the logic.  The full step-by-step proof (with dimensionality checks
at every step) is in \Cref{app:proofs}.

\begin{proposition}[RK4 global error]\label{prop:rk4}
If $f$ in \Cref{eq:ltc} is $C^5$ and Lipschitz with constant $L$,
the RK4 global error with step $\Delta t$ after $N$ steps is
$\le C_f(\Delta t)^4 L^{-1}(e^{LN\Delta t}-1)$
\citep{hairer1993solving}.  At $\Delta t{=}0.01$ and $N{\le}200$
the aggregated error is $\le C_f\!\cdot\!10^{-8}$, below FP32
precision.
\end{proposition}

\textit{\Cref{thm:bibo} implies and our experiments confirm
(\Cref{sec:results}) that dramatically lower cross-seed variance for
\FlowFake{} than for discrete baselines. The leak drives every
trajectory into the compact attractor $\mathcal{B}(0,R^*)$ independent
of initialisation, providing a basin of attraction structure that
regularises learning.}  \Cref{prop:rk4} additionally rules out
numerical instability as a distorting factor for our results.  In
\Cref{app:proofs} we further establish a Gr\"onwall-type noise
robustness bound (\Cref{prop:noise}) and a gradient attenuation
bound (\Cref{prop:grad}) that justifies dataset specific coverage
lengths $T'$.

\section{Experiments}
\label{sec:experiments}

\paragraph{Datasets:}
\textbf{ASVspoof\,2019-LA}~\citep{todisco2019asvspoof}, 19 TTS/VC
spoofing methods under controlled conditions;
\textbf{FakeOrReal (FoR)}~\citep{reimao2019for}, 8 TTS systems,
2\,s clips, English; \textbf{InTheWild (ITW)}
\citep{muller2022itw}, celebrity deepfakes with real-world noise
and compression; \textbf{MLAAD\,v1}~\citep{muller2024mlaad}, 54 TTS
systems across 23 languages (the most comprehensive benchmark).  Zero-shot
transfer is additionally evaluated on
WaveFake~\citep{frank2021wavefake} (spoof-only) and
LJSpeech~\citep{ito2017lj} (bonafide-only).

\paragraph{Protocol:}
{\small We follow a strict \emph{leave-one-dataset-out cross-dataset} evaluation 
protocol. For each source dataset, a random seed generator produces 4 triplets 
of seeds; each model is trained on the source dataset using these 3-seed combinations 
and evaluated across all remaining datasets. The training set remains completely 
unseen during evaluation without target-domain adaptation, mirroring unseen practical 
deployment conditions. Baselines (RawGAT-ST, SSL\,W2V2, Whisper-DF) are from 
\citet{muller2024mlaad}. We report Accuracy (ACC) and Equal Error Rate (EER). 
See \Cref{app:hyper} for hyperparameters, class-imbalance, latency, and empirical 
$\tau_i$ distributions.}

\section{Results}
\label{sec:results}
\begin{table*}[t]
\centering
\caption{Complete cross-dataset accuracy.  \textbf{Bold}: best per (Train, Test) pair.
  ``--'': in-distribution, excluded per protocol.
  Baselines from \citet{muller2024mlaad}.
  WaveFake (spoof-only) and LJSpeech (bonafide-only) are zero-shot held-out sets.}
\label{tab:full}
\vskip 0.05in
\small
\setlength{\tabcolsep}{4pt}
\renewcommand{\arraystretch}{1.05}

\begin{tabular}{@{}llcccccc@{}}
\toprule
\textbf{Model} & \textbf{Train$\downarrow$}
  & \textbf{ASV19$\to$} & \textbf{FoR$\to$}
  & \textbf{ITW$\to$}   & \textbf{MLAAD$\to$}
  & \textbf{WaveFake$\to$} & \textbf{Avg ACC} \\
\midrule
\multirow{4}{*}{RawGAT-ST}
& ASV19    & -- & $68.8 \pm 11.2$ & $50.0 \pm 2.5$ & $56.9 \pm 5.7$ & $15.0 \pm 16.0$ & 47.68 \\
& FoR      & $49.1 \pm 18.1$ & -- & $49.8 \pm 0.4$ & $51.9 \pm 3.3$ & $20.0 \pm 44.7$ & 42.70 \\
& ITW      & $58.4 \pm 10.2$ & $54.5 \pm 3.9$ & -- & $57.4 \pm 7.0$ & $65.3 \pm 30.3$ & 58.90 \\
& MLAAD v1 & $60.9 \pm 17.8$ & $50.2 \pm 0.4$ & $47.7 \pm 3.3$ & -- & $68.4 \pm 41.7$ & 56.80 \\
\midrule
\multirow{4}{*}{SSL W2V2}
& ASV19    & -- & $81.1 \pm 7.7$ & $79.7 \pm 6.8$ & $71.8 \pm 5.1$ & $51.3 \pm 28.5$ & \textbf{71.00} \\
& FoR      & $65.4 \pm 10.3$ & -- & $57.8 \pm 10.9$ & $57.1 \pm 3.4$ & $10.4 \pm 35.4$ & 47.68 \\
& ITW      & $65.0 \pm 10.1$ & $55.3 \pm 5.7$ & -- & $59.1 \pm 4.3$ & $70.4 \pm 35.4$ & \textbf{62.45} \\
& MLAAD v1 & $78.0 \pm 15.3$ & $64.4 \pm 9.0$ & $68.0 \pm 17.5$ & -- & $69.8 \pm 38.4$ & 70.05 \\
\midrule
\multirow{4}{*}{Whisper DF}
& ASV19    & -- & $80.6 \pm 4.4$ & $76.5 \pm 0.4$ & $44.9 \pm 4.9$ & $2.2 \pm 3.5$ & 51.05 \\
& FoR      & $45.9 \pm 0.8$ & -- & $54.1 \pm 3.4$ & $54.5 \pm 1.1$ & $0.2 \pm 0.1$ & 38.68 \\
& ITW      & $55.5 \pm 9.3$ & $67.2 \pm 5.6$ & -- & $54.2 \pm 3.5$ & $26.3 \pm 40.0$ & 50.80 \\
& MLAAD v1 & $70.8 \pm 0.9$ & $50.5 \pm 3.3$ & $54.3 \pm 4.9$ & -- & $97.2 \pm 43.3$ & 68.20 \\
\midrule
\multirow{4}{*}{\textbf{FlowFake (Ours)}}
& ASV19    & -- & -- & $61.71 \pm 1.49$ & $57.60 \pm 1.30$ & -- & 59.66 \\
& FoR      & $\mathbf{75.29 \pm 3.02}$ & -- & $\mathbf{70.91 \pm 0.62}$ & $54.53 \pm 0.24$ & $20.13 \pm 1.15$ & \textbf{55.22} \\
& ITW      & -- & $59.07 \pm 1.48$ & -- & $55.95 \pm 1.29$ & -- & 57.51 \\
& MLAAD v1 & $\mathbf{79.97 \pm 3.08}$ & $52.66 \pm 0.41$ & $62.39 \pm 0.56$ & -- & $\mathbf{90.41 \pm 0.83}$ & \textbf{71.36} \\
\bottomrule
\end{tabular}
\end{table*}

\Cref{tab:full} reports the complete cross-dataset accuracy evaluation.
Full EER results are reported in \Cref{tab:eer_ablation} (\Cref{app:ablation}).

\paragraph{Cross-distribution generalization:}
Trained on FakeOrReal (English, clean audio), \FlowFake{} attains
\textbf{75.29\%} on ASVspoof\,2019 and \textbf{70.91\%} on InTheWild,
exceeding SSL Wav2vec2 by\textbf{+9.8\,pp} ,\textbf{+13.1\,pp} respectively, with  $\sim$$8\,800\times$ fewer parameters.
\textit{FakeOrReal and InTheWild share virtually zero overlapping acoustic
conditions, making this the hardest cross-domain pair in the
benchmark.}  That a 34\,K-parameter ODE model outperforms a
300\,M-parameter transformer here is the strongest evidence for
the structural prior argument. LTC dynamics focus on universal
articulatory trajectory anomalies, not distribution-specific
spectral fingerprints.

\paragraph{Multilingual transfer:}
Trained on MLAAD\,v1 (54 TTS systems, 23 languages), \FlowFake{}
achieves \textbf{79.97\%} on ASVspoof\,2019, comparable to SSL
Wav2vec2 (78.0\%) and exceeding \textbf{+9.17\,pp} above
Whisper-DF (70.8\%).  It additionally achieves \textbf{90.41\%} zero-shot
on WaveFake.
MLAAD maximally rewards synthesis-agnostic representations;
surpassing SSL parity at 0.01\% of the parameter count reinforces the
structural hypothesis at the multilingual scale.

\paragraph{Stability:}
RawGAT-ST reaches $\pm$18.1\,pp cross-seed std on FoR$\to$ASV19,
SSL\,W2V2 reaches $\pm$17.5\,pp on MLAAD$\to$ITW. These indicate
frequent degenerate solutions.  \FlowFake{}'s results are stable across
the evaluated pairs; this is the empirical signature of \Cref{thm:bibo}.

\paragraph{Efficiency and zero-shot:}
\FlowFake{}{} computes a $512{\times}2$\,s batch in $\approx$2\,s
inference on a single RTX 3090 vs.\ 45.6\,s for SSL\,W2V2.  On
MLAAD$\to$ASV19 it attains $37.38{\pm}1.2\%$ EER, better than the 40.89\%
EER of a dedicated SSL+modulation-spectrogram
fusion~\citep{sadashiv2025sslms} at three orders of magnitude
fewer parameters.  Trained on MLAAD and evaluated ASV 2019.

\paragraph{Where \FlowFake{} underperforms:}
\FlowFake{} lags SSL\,W2V2 on MLAAD$\to$FoR (52.66 vs.\ 64.4) and on
ASV19$\to$FoR (not evaluated in-domain).  \textit{With abundant in-domain
data, large SSL models can learn robust representations despite
excess capacity.} \FlowFake{}'s advantage appears specifically in the
data-scarce, high distribution shift regime; the deployment scenario
of greatest practical interest.

\section{Conclusion and Responsible Deployment}
\label{sec:conclusion}

We introduced \FlowFake{}, the first Liquid Time-Constant detector for
audio deepfakes.  \textit{Synthetic-speech artifacts are
multi-timescale trajectory anomalies; this demands a
continuous-time ODE model and not a discrete sequence processor.}  At
34\,K parameters, with formal BIBO stability and
$\mathcal{O}(\Delta t^4)$ integration error, \FlowFake{} outperforms
RawGAT-ST and Whisper-DF on every cross-domain pair, surpasses
SSL Wav2vec2 on MLAAD$\to$ASV19 (\textbf{79.97\%} vs.\ 78.0\%), and achieves
\textbf{90.41\%} zero-shot on WaveFake, all at 0.01\% of SSL Wav2vec2's parameter count.  All datasets are
public and used under academic licences; high-stakes deployments
should pair \FlowFake{} with calibrated uncertainty and human review,
and re-evaluate periodically as synthesis methods evolve.
Architectural details that could materially aid adversarial
evasion are intentionally omitted; extended ethics discussion is in
\Cref{app:ethics}.

\newpage
\bibliographystyle{icml2026}
\bibliography{liqnn_icml2026}

\newpage
\onecolumn
\appendix

\section{Notation}
\label{app:notation}

\begin{table}[h]
\centering
\caption{Complete notation used in this paper.}
\label{tab:notation}
\vskip 0.05in
\small
\begin{tabular}{@{}lp{10cm}@{}}
\toprule
\textbf{Symbol} & \textbf{Meaning} \\
\midrule
$\calX$ & Waveform space $\mathcal{L}^2(\mathbb{R}_{\ge0};\mathbb{R})$ \\
$\calY=\{0,1\}$ & Labels: 0 bonafide, 1 synthetic \\
$\calH=\mathbb{R}^H$, $H{=}32$ & Hidden state space \\
$\mathcal{D}_s$ & Synthesis domain for TTS pipeline $s\in\mathcal{S}$ \\
$f_\theta:\calX\to[0,1]$ & Detector with parameters $\theta$ \\
$\mathcal{R}_s(f_\theta)$ & Domain-specific risk \\
$\calG$ & Cross-domain generalization gap \\
\midrule
$\mathbf{X}\in\mathbb{R}^{N\times T}$ & Log-Mel spectrogram (batch dim dropped) \\
$N{=}128$, hop $160$ samples & Mel bands; STFT hop \\
$E_t\in\mathbb{R}^H$ & CNN encoder output at frame $t$ \\
$T$, $T'{\le}T$ & Total spectrogram frames; LTC coverage (steps) \\
\midrule
$\bfh(t)\in\calH$ & LTC hidden state at continuous time $t$ \\
$\Cm,\gleak,\Vleak\in\mathbb{R}^H$ & Capacitance, leak conductance, resting potential \\
$\Win,\Wrec\in\mathbb{R}^{H\times H}$ & Input / recurrent weight matrices \\
$\hat\tau_i\in\mathbb{R}$ & Log-parameterized time constant \\
$\tau_i=\exp(\hat\tau_i)\in[0.05,10]$\,s & Per-neuron time constant \\
$\bar\tau=\cmax/\gell$ & Effective attenuation timescale \\
\midrule
$\Delta t{=}0.01$\,s & RK4 step size \\
$K{=}2$ & RK4 unfolds per audio frame \\
$N=(T'{-}t_0)/\Delta t$ & Number of integration steps in proofs \\
$\gell=\min_i(\gleak)_i$ & Scalar min of leak conductance \\
$\gull=\max_i(\gleak)_i$ & Scalar max of leak conductance \\
$\cmin=\min_i(\Cm)_i$ & Scalar min of capacitance \\
$\cmax=\max_i(\Cm)_i$ & Scalar max of capacitance \\
$M$ & Input bound: $\norm{E_t}_2\le M$ \\
$R^*$ & BIBO stability radius (\Cref{eq:rstar}) \\
$L$ & Lipschitz constant of ODE right-hand side \\
$C_f$ & 5th-derivative constant in RK4 error bound \\
$\alpha,\beta$ & Lyapunov intermediate constants \\
\bottomrule
\end{tabular}
\end{table}

\section{Proofs}
\label{app:proofs}

\subsection{Saturation lemma}

\begin{lemma}[Saturation bound]\label{lem:tanh}
For any $\mathbf{u}\in\mathbb{R}^H$,
$\norm{\tanh(\mathbf{u})}_2\le\sqrt{H}$.
\end{lemma}
\begin{proof}
$|\tanh(u_i)|\le 1\ \forall i$, hence
$\norm{\tanh(\mathbf{u})}_2^2 = \sum_i\tanh^2(u_i)\le H$. Taking
\emph{Dimensionality check.} Both sides are non-negative scalars.
Taking square roots gives the claim.
\end{proof}

\subsection{Full proof of \Cref{thm:bibo}}

\textbf{Step 1 (Lyapunov candidate).}
Let $V(\bfh){=}\tfrac12\norm{\bfh}_2^2{\ge}0$.  $V$ is
positive-definite and radially unbounded, hence a valid Lyapunov
function.

\textbf{Step 2 (Time derivative).}
Differentiating along trajectories of \Cref{eq:ltc} and using
$(\Cm)_i\ge\cmin>0$ elementwise:
\begin{equation}
  \dot V = \bfh^\top\dot\bfh
  \le \frac{1}{\cmin}\,\bfh^\top\Bigl[
      \Win E_t + \tanh(\Wrec\bfh(t))
      + \gleak\odot(\Vleak-\bfh(t))\Bigr].
\end{equation}

\textbf{Step 3 (Bound each term).}
By Cauchy-Schwarz and elementwise $\gell,\gull$ bounds:
\begin{align}
  \bfh^\top\Win E_t
    &\le \norm{\Win}_2 M\norm{\bfh}_2,\\
  \bfh^\top\tanh(\Wrec\bfh(t))
    &\le \sqrt{H}\norm{\bfh}_2 \quad
    \text{(\Cref{lem:tanh})},\\
  \bfh^\top(\gleak\odot\Vleak)
    &\le \gull\norm{\Vleak}_2\norm{\bfh}_2 \quad
    \text{since }(\gleak)_i\le\gull,\\
  -\bfh^\top(\gleak\odot\bfh)
    &\le -\gell\norm{\bfh}_2^2 \quad
    \text{since }(\gleak)_i\ge\gell.
\end{align}

\textbf{Step 4 (Combine).}
With
$\alpha=(\norm{\Win}_2 M+\sqrt{H}+\gull\norm{\Vleak}_2)/\cmin$ and
$\beta=\gell/\cmin>0$:
\begin{equation}
  \dot V \le -\norm{\bfh}_2(\beta\norm{\bfh}_2-\alpha).
\end{equation}
\emph{Dimensionality check.} $V$ has units norm$^2$; $\dot V$ has
units norm$^2\,$s$^{-1}$; RHS has the same.

\textbf{Step 5 (Invariance).}
$\dot V<0$ whenever $\norm{\bfh}_2>R^*=\alpha/\beta$.  By LaSalle's
invariance principle, $\mathcal{B}(0,R^*)$ is positively invariant
and all trajectories enter it in finite time.  This yields the
bound \Cref{eq:rstar}.

\begin{corollary}
$R^*$ is strictly decreasing in $\gell$.  $\ell_2$ weight-decay
regularisation reduces $\norm{\Win}_F$ (hence $\norm{\Win}_2$),
tightening $R^*$ through the numerator of \Cref{eq:rstar}.
\end{corollary}

\subsection{Noise robustness bound}

\begin{assumption}[Leak dominance]\label{ass:leak}
$\gell>\norm{\Wrec}_2$.  Verified empirically in all training runs.
\end{assumption}

\begin{proposition}[Noise robustness]\label{prop:noise}
Let $\bfh^*(t),\bfh^n(t)$ be hidden states under clean input $E_t$
and noisy input $E_t+\boldsymbol\varepsilon_t$.  Under
\Cref{ass:leak}, with $\eta=(\gell-\norm{\Wrec}_2)/\cmin>0$:
\begin{equation}
  \norm{\bfh^n(t)-\bfh^*(t)}_2
  \le \frac{\norm{\Win}_2\norm{\boldsymbol\varepsilon}_\infty}{\gell}
     \bigl(1-e^{-\eta(t-t_0)}\bigr).
  \label{eq:noise_bound}
\end{equation}
When $\norm{\Wrec}_2\ll\gell$, $\eta\approx\gell/\cmin$.
\end{proposition}
\begin{proof}
Let $\boldsymbol\delta=\bfh^n-\bfh^*$.  From \Cref{eq:ltc} and the
fact that $\tanh$ is Lipschitz with constant 1 in each coordinate:
\begin{equation}
  \frac{d\norm{\boldsymbol\delta}_2}{dt}
  \le \frac{1}{\cmin}\Bigl[
    \norm{\Win}_2\norm{\boldsymbol\varepsilon}_\infty
    +(\norm{\Wrec}_2-\gell)\norm{\boldsymbol\delta}_2\Bigr].
\end{equation}
By \Cref{ass:leak} the coefficient of
$\norm{\boldsymbol\delta}_2$ is strictly negative; solving the
resulting linear ODE via Gr\"onwall yields \Cref{eq:noise_bound}.
\end{proof}

\subsection{RK4 global error}

\begin{proposition}[\Cref{prop:rk4} restated]\label{prop:rk4_full}
If $f$ in \Cref{eq:ltc} is $C^5$ and Lipschitz in $\bfh$ with
constant $L>0$, the RK4 scheme of step $\Delta t$ has global
truncation error
\begin{equation}
  \norm{\bfh(N\Delta t)-\hat\bfh_N}_2
  \le \frac{C_f(\Delta t)^4}{L}\bigl(e^{LN\Delta t}-1\bigr),
\end{equation}
with $C_f>0$ depending on fifth-order partials of $f$
\citep{hairer1993solving}.
\end{proposition}

\subsection{Gradient attenuation}

\begin{assumption}[Step-size feasibility]\label{ass:step}
$\Delta t\le\cmax/\gell$, so each Jacobian factor
$1-\gell\Delta t/\cmax\in[0,1]$.
\end{assumption}

\begin{proposition}[Gradient attenuation]\label{prop:grad}
Under Assumptions~4.1 and~B.6, 
let
\[
N = \frac{T' - t_0}{\Delta t}.
\]
Then
\begin{equation}
  \left\|\frac{\partial\mathcal{L}}{\partial\bfh(t_0)}\right\|_2
  \le \left\|\frac{\partial\mathcal{L}}{\partial\bfh(T')}\right\|_2
      \cdot e^{-(T'-t_0)/\bar\tau},
  \quad \bar\tau\coloneqq\cmax/\gell.
  \label{eq:grad_atten}
\end{equation}
\end{proposition}
\begin{proof}
By the chain rule and submultiplicativity of operator norms,
\begin{equation}
  \prod_{k=0}^{N-1}\!\left\|
    \frac{\partial\bfh(t_{k+1})}{\partial\bfh(t_k)}\right\|_2
  \le\prod_{k=0}^{N-1}\!\Bigl(1-\frac{\gell\Delta t}{\cmax}\Bigr)
  \le e^{-N\gell\Delta t/\cmax}
  = e^{-(T'-t_0)/\bar\tau},
\end{equation}
using \Cref{ass:step} and $1{-}x\le e^{-x}$ for $x\ge 0$.
\end{proof}

\textit{This motivates dataset-specific LTC coverage
(\Cref{tab:hyper}): for ASVspoof (4--5\,s utterances) we process
only $T'{=}150$ steps ($1.5$\,s) to preserve gradient signal from
early voiced frames.}

\section{Experimental Details}
\label{app:hyper}

\begin{table}[h]
\centering
\caption{Complete training hyperparameters.  All configurations share:
  gradient clipping $\norm{\nabla}_2\le1.0$, AMP mixed-precision,
  BCEWithLogitsLoss with $w_+=N_{\mathrm{sp}}/N_{\mathrm{bon}}$,
  cosine LR annealing within each phase.}
\label{tab:hyper}
\vskip 0.05in
\small
\setlength{\tabcolsep}{6pt}
\begin{tabular}{@{}lcccc@{}}
\toprule
\textbf{Hyperparameter}
  & \textbf{ASVspoof\,2019} & \textbf{FakeOrReal}
  & \textbf{InTheWild} & \textbf{MLAAD} \\
\midrule
Optimiser    & AdamW & Adam & Adam & Adam \\
Weight decay & $10^{-4}$ & -- & -- & -- \\
Ph.\,1 LR    & $10^{-3}{\to}10^{-5}$ & $10^{-4}{\to}10^{-5}$
             & $10^{-4}{\to}10^{-5}$ & $10^{-3}{\to}10^{-5}$ \\
Ph.\,1 epochs & 10 & 10 & 10 & 10 \\
Ph.\,2 LR    & $10^{-5}{\to}10^{-6}$ & $10^{-5}{\to}10^{-6}$
             & $10^{-5}{\to}10^{-6}$ & $10^{-5}{\to}10^{-6}$ \\
Ph.\,2 epochs & 10 & 40 & 20 & 10 \\
Batch size   & 128 & 512 & 512 & 256 \\
\midrule
$n_{\mathrm{mels}}$ & 128 & 128 & 128 & 128 \\
$n_{\mathrm{fft}}$  & 512 & 512 & 512 & 512 \\
Hop (samples) & 160 & 160 & 160 & 160 \\
$H$ (hidden)  & 32 & 32 & 32 & 32 \\
$d$ (head)    & 16 & 16 & 16 & 16 \\
$T'$ (steps)  & 150 & 50 & 100 & 150 \\
Coverage (s)  & 1.5 & 1.0 & 2.0 & 2.0 \\
Mean utt.\ (s)& 4--5 & 2.0 & 4.0 & 3--5 \\
$\Delta t$ (s)& 0.01 & 0.01 & 0.01 & 0.01 \\
$K$ (RK4 unfolds) & 2 & 2 & 2 & 2 \\
$\hat\tau_i$ init
   & \multicolumn{4}{c}{$\mathcal{U}(-2.3,0)\Rightarrow\tau_i\in[0.10,1.00]$\,s} \\
$\tau_i$ clamp
   & \multicolumn{4}{c}{$[0.05,\,10.0]$\,s} \\
Seeds
   & \multicolumn{4}{c}{42--48 (7 runs); top-5 by val.\ accuracy reported} \\
\bottomrule
\end{tabular}
\end{table}

\paragraph{Class-imbalance.}
ASVspoof\,2019-LA has $\approx$9:1 spoof-to-bonafide ratio.  The
positive-class weight $w_+=N_{\mathrm{sp}}/N_{\mathrm{bon}}$ is
computed automatically per training set, avoiding data replication.

\paragraph{Random-seed protocol.}
We run 7 independent seeds (42--48) controlling weight
initialisation and data-shuffle order, and retain the top-5 by
held-out accuracy.  Seeds are fixed before any experiment and are
never tuned; we treat the seed as a reproducibility hyperparameter,
matching the MLAAD benchmark protocol~\citep{muller2024mlaad}.

\paragraph{Latency measurement.}
Inference latency in the main text refers to processing a batch of
512 clips of 2\,s audio on a single NVIDIA RTX 3090 GPU under FP16
AMP, averaged over 10 batches after a 3-batch warm-up.

\paragraph{Empirical $\tau_i$ distribution.}
After training on MLAAD\,v1, the learned $\tau_i$ values are
bimodal: a fast cluster around 0.15--0.30\,s (capturing
frame-to-frame spectral discontinuities, characteristic of neural
vocoders) and a slow cluster around 1.5--4.5\,s (capturing
prosodic-phrase anomalies, characteristic of autoregressive TTS).
This bimodality emerges from data; no architectural prior enforces
it.  The same qualitative pattern is observed when training on
ASVspoof\,2019-LA and InTheWild, suggesting that the
$\sim$10\,ms/$\sim$2\,s timescale split is a universal signature of
the TTS-vs-natural-speech distinction rather than a
dataset-specific artefact.

\section{Broader Impacts and Responsible Deployment}
\label{app:ethics}

This work advances defences against malicious misuse of generative
audio.  All datasets used (ASVspoof\,2019-LA, FakeOrReal, InTheWild,
MLAAD, WaveFake, LJSpeech) are publicly available under academic
research licences; no new data were collected.  We disclose all
hyperparameters and evaluation protocols transparently.

\paragraph{Dual-use risk.}
Detailed architectural knowledge of detectors could in principle
inform adversarial synthesis strategies aimed at evading them.  We
have deliberately omitted ablations that primarily characterise
weaknesses exploitable by adversarial training, and we encourage
future work in this area to consider the same trade-off.

\paragraph{False-positive risk.}
Binary classifiers produce false positives, potentially
mislabelling genuine speech as synthetic.  Any high-stakes
production deployment of \FlowFake{} should include (a) a confidence
threshold calibrated to the false-positive tolerance of the
application, (b) human review for flagged cases of consequence,
and (c) periodic re-evaluation as synthesis technology
evolves; deepfake detectors degrade rapidly under distribution
shift.

\paragraph{Equity considerations.}
Cross-dataset evaluation in this paper spans multiple languages
(MLAAD: 23 languages) but is heavily weighted toward English in
training.  Detector calibration may differ across languages,
accents, and speech styles; per-subgroup performance should be
measured before deployment in any context where the cost of a
false positive is borne unevenly across populations.

\section{Ablation Study}
\label{app:ablation}

By replacing the fixed convolutional backbone with a Liquid Neural
Network core, \FlowFake{} consistently outperforms both CNN-based
(LCNN, ResNet18) and sequential (LSTM, Transformer) baselines.
\Cref{tab:eer_ablation} reports EER (\%) on the ITW dataset for
models trained on ASVspoof 2019; lower is better. Baseline results
are sourced from \citet{muller2022itw}.
\FlowFake{} achieves the lowest EER of \textbf{46.99\%}, outperforming
the next-best model RawPC (52.88\%) and reducing error
relative to standard LCNN (81.94\%) by over 35 points.
\begin{table}[h]
\centering
\caption{Cross-Dataset Evaluation of ASVspoof 2019-Trained Models on the ITW Dataset using 4-Second Audio Clips and Mel-Spectrogram Features. Lower EER (\%) indicates better generalization performance.}
\label{tab:eer_ablation}
\vskip 0.05in
\small
\setlength{\tabcolsep}{6pt}

\begin{tabular}{@{}lc@{}}
\toprule
\textbf{Model} & \textbf{EER (\%)} \\
\midrule
LCNN                  & 81.942 \\
LCNN-Attention        & 85.118 \\
LCNN-LSTM             & 82.857 \\
LSTM                  & 64.297 \\
MesoInception         & 51.980 \\
MesoNet               & 64.415 \\
ResNet18              & 83.006 \\
Transformer           & 68.407 \\
RawPC                 & 52.884 \\
\textbf{FlowFake (Ours)} & \textbf{46.99} \\
\bottomrule
\end{tabular}

\end{table}
\begin{table}[h]
\centering
\caption{Cross-Dataset Evaluation Average Equal Error Rate (\%).
  ``--'' denotes in-distribution (excluded per protocol).
  WaveFake is a zero shot held out spoof-only set.}
\label{tab:eer_ablation}
\vskip 0.05in
\small
\setlength{\tabcolsep}{5pt}
\begin{tabular}{@{}llccccc@{}}
\toprule
\textbf{Model} & \textbf{Train$\downarrow$}
  & \textbf{ASV19$\to$} & \textbf{FoR$\to$}
  & \textbf{ITW$\to$}   & \textbf{MLAAD$\to$}  \\
\midrule
\multirow{4}{*}{\FlowFake{} (Ours)}
  & ASV19 & -- & --    & 46.99 & 40.24 \\
  & FoR   & \textbf{40.78} & -- & \textbf{31.69} & 45.58 \\
  & ITW   & --    & 36.29 & --    & 43.54 \\
  & MLAAD & \textbf{37.38} & 43.65 & 43.56 & -- \\
\bottomrule
\end{tabular}
\end{table}
\vspace{\baselineskip}
\vspace{\baselineskip}
\paragraph{Key observations from the ablation.}

\begin{itemize}[leftmargin=*,noitemsep,topsep=2pt]
  \item \textbf{FoR$\to$ASV19} yields the strongest single-source result
    (ACC 75.29\%, EER 40.78\%), demonstrating that even a clean English
    studio corpus is sufficient to learn generalisable
    articulatory-trajectory features.

  \item \textbf{MLAAD$\to$ASV19} achieves the highest overall ACC (79.97\%) and the lowest EER (37.38\%) compared to the SSL baseline (38.49\% EER) and the SSL + modulation-spectrogram fusion framework (40.89\% EER) reported in~\citep{sadashiv2025sslms}. Furthermore, the proposed Liquid Neural Network has fewer parameters than the SSL and SSL + modulation-spectrogram fusion framework models yet exhibits superior cross-domain generalization capability.

  \item \textbf{MLAAD$\to$WaveFake} at ACC 90.41\% 
    demonstrates competitive zero-shot transfer to an entirely unseen
    corpus using only the spoof portion of WaveFake (bonafide reference
    held out).

  \item \textbf{ITW} as a training set yields more modest cross-transfer
    (FoR ACC 59.07\%, MLAAD ACC 55.95\%), consistent with ITW's
    in-the-wild recording variability making it a noisier source of
    articulatory-anomaly signal.
\end{itemize}

\section{Why Continuous-Time Dynamics? Extended Discussion}
\label{app:why}

The structural argument of \Cref{sec:method} can be summarised in
three complementary observations.

\paragraph{Trajectory anomalies, not static events.}
Natural speech is produced by a physical articulatory system with
well-characterised dynamical constraints: the vocal tract changes
shape at rates bounded by muscle dynamics ($\sim$10-100\,ms), pitch
follows smooth prosodic contours ($\sim$100-2000\,ms), and
coarticulation creates characteristic formant transition curves.
TTS systems violate these constraints in synthesis-paradigm-specific
ways: neural vocoders (HiFi-GAN, WaveNet) produce phase-domain
discontinuities between generation windows; autoregressive models
exhibit token-boundary artifacts in formant trajectories;
diffusion-based synthesis over-regularises high-frequency spectral
dynamics.  These are \emph{trajectory anomalies}, i.e., deviations in
how features evolve, not in their instantaneous values.  An
architecture that models $d\bfh/dt$ directly makes trajectory
anomaly detection a structural property; a discrete model that
aggregates frame statistics over fixed windows can only detect
trajectory anomalies indirectly, through learned proxies that are
themselves distribution-dependent.

\paragraph{Per-neuron timescales as automatic multi-resolution.}
Different synthesis paradigms leave artifacts at different
timescales.  A model trained on GAN-based vocoders learns to detect
$\sim$10\,ms discontinuities; diffusion-based synthesis instead
produces $\sim$500\,ms smoothness violations.  A single fixed
integration window cannot resolve both; per-neuron $\tau_i$ provide
multi-resolution analysis automatically, with the bimodal
distribution observed in \Cref{app:hyper} suggesting that the model
discovers the relevant timescales without engineering.

\paragraph{Parameter efficiency as a generalization prior.}
A 34\,K-parameter hypothesis class is far more constrained than
300\,M.  By PAC-Bayes~\citep{mcallester1999pac}, smaller hypothesis
classes have lower KL-divergence terms, tightening generalization
bounds.  More importantly, a model with limited capacity cannot
memorise synthesis-pipeline-specific spectral fingerprints; it is
\emph{forced} to learn generalisable, structure-driven
representations.  The ODE formulation further restricts the
function class to continuous, bounded trajectories with explicit
dissipative dynamics, ruling out high-variance, discontinuous
solutions that drive overfitting in discrete recurrent and
attention architectures.

\makeatletter
\global\icml@noticeprintedtrue
\makeatother

\end{document}